\documentclass[letter]{IEEEtran}
\usepackage{graphicx,color}
\usepackage{cite}
\usepackage{setspace} 
\usepackage{amsmath,amsthm,amssymb}
\usepackage{multirow}
\usepackage{hhline}
\usepackage{epsfig}
\usepackage{epstopdf}
\usepackage{verbatim}
\usepackage{algorithm}
\usepackage{algpseudocode}
\usepackage{cases}
\usepackage{array}
\usepackage{subfigure}

\newtheorem{lemma}{Lemma}

\allowdisplaybreaks

\usepackage{forloop}
\newcounter{ct}

\makeatletter 
\def\@eqnnum{{\normalsize \normalcolor (\theequation)}} 
\makeatother
\begin{document}
\title{Joint Optimization of IRS Deployment and Passive Beamforming to Enhance the Received Power}
\author{Jyotsna Rani, \textit{Student Member, IEEE}, Deepak Mishra, \textit{Member, IEEE}, Ganesh Prasad, \textit{Member, IEEE}, Ashraf Hossain, \textit{Senior Member, IEEE}, and Swades De, \textit{Senior Member, IEEE}, Kuntal Deka, \textit{Member, IEEE}
	\thanks{J. Rani and K. Deka are with the Department of Electronics and Electrical Engineering, Indian Institute of Technology, Guwahati, India (e-mail: \{r.jyotsna, kuntaldeka\}@iitg.ac.in).

		G. Prasad and A. Hossain are with the Department of Electronics and Communication Engineering, National Institute of Technology, Silchar, India (e-mail: \{gpkeshri, ashraf\}@ece.nits.ac.in).
		
		D. Mishra is with the School of Electrical Engineering and Telecommunications
		UNSW Sydney, Australia (e-mail: d.mishra@unsw.edu.au).
		
		S. De is with the Department of Electrical Engineering and Bharti School of Telecommun., Indian Institute of Technology Delhi, New Delhi, India (e-mail: swadesd@ee.iitd.ac.in).
}}

\maketitle
\begin{abstract}
Intelligent reflecting surface (IRS) has recently emerged as a promising technology for beyond fifth-generation (B5G) and 6G networks conceived from metamaterials that smartly tunes the signal reflections via a large number of low-cost passive reflecting elements. However, the IRS-assisted communication model and the optimization of available resources needs to be improved further for more efficient communications. This paper investigates the enhancement of received power at the user end in an IRS assisted wireless communication by jointly optimizing the phase shifts at the IRS elements and its location. Employing the conventional Friss transmission model, the relationship between the transmitted power and reflected power is established. The expression of received power incorporates the free space loss, reflection loss factor, physical dimension of the IRS panel, and radiation pattern of the transmit signal. Also, the expression of reflection coefficient of IRS panel is obtained by exploiting the existing data of radar communications. Initially exploring a single IRS element within a two-ray reflection model, we extend it to a more complex multi-ray reflection model with multiple IRS elements in 3D Cartesian space. The received power expression is derived in a more tractable form, then, it is maximized by jointly optimizing the underlying underlying variables, the IRS location and the phase shifts. To realize the joint optimization of underlying variables, first, the phase shifts of the IRS elements are optimized to achieve constructive interference of received signal components at the user. Subsequently, the location of the IRS is optimized at the obtained optimal phase shifts. Numerical insights and performance comparison reveal that joint optimization leads to a substantial $37\%$ enhancement in received power compared to the closest competitive benchmark scheme.
\end{abstract} 

\section{Introduction}
Intelligent Reflecting Surfaces (IRS) are cutting-edge wireless communication technologies that hold tremendous promise in transforming wireless networks. IRS consists of specially engineered materials integrated with antennas or passive elements. They can manipulate and control the propagation of electromagnetic waves in real-time, allowing for dynamic and adaptive wireless signal shaping. Optimizing the available resources in the IRS-assisted networks can lead to an improved signal strength, minimum interference, and a larger coverage range. IRS technology is expected to revolutionize various applications, such as 5G and beyond, Internet of Things (IoT), and wireless communications in challenging environments. 

\subsection{Related Works}
Unlike typical Radio Frequency (RF) receivers, IRS lacks the capability to detect the received signal~\cite{wu23}, which presents a challenge in modeling or estimating the associated channels. The existing works considered the conventional channel models for IRS links like Rayleigh fading~\cite{bas24} or Rician fading for deterministic line-of-sight (LOS) components and i.i.d complex Gaussian distribution for non-line-of-sight (NLOS) components~\cite{han25}. For IRS-assisted multiple-input single-output (MISO) systems, in \cite{dang27}, the channel modeled, and a control loop between the base station and IRS panel used for channel estimation using minimum mean squared error (MMSE). While in an another approach, stochastic channel model based on geometry was taken where the correlation of different subchannels of IRS elements was considered. However, to fully explore IRS communication, it is crucial to incorporate factors such as free space loss, reflection loss factor, physical dimensions of the IRS panel, and the radiation pattern of the transmit signal.

Another line of research focused on optimizing available resources in IRS-assisted wireless networks. In \cite{wu28}, joint optimization of discrete phase shifts of the IRS panel and transmit beamforming minimizes the transmit power of the access point (AP). In \cite{wang29}, the relationship explored between user's transmit power and phase shift at the IRS that leads to joint optimization of transmit power and phase shift for overall power reduction. The authors in \cite{xu30} jointly optimized the phase shifts of IRS, downlink transmit beamforming, and artificial noise covariance matrix at the base station to maximize the sum secrecy rate of the network. Regarding location optimization of IRS, limited works investigated it in \cite{zhang31,jiao33}. In \cite{zhang31}, distributed and centralized deployment strategies with fixed IRS locations were analyzed. Conversely, in \cite{jiao33}, the IRS location optimized initially, thereafter, the transmit beamforming and phase shifts optimized alternately using an iterative algorithm. 

Besides, many works focused on joint optimization of system parameters to enhance the performance of IRS-assisted networks. In \cite{ge12}, active and passive beamforming at an unmanned aerial vehicle (UAV) and the IRS trajectory were jointly optimized to maximize user's received power. The phase shifts of the IRS along with base station (BS) beamforming and artificial noise covariance matrix were optimized in \cite{xu13} to maximize secrecy rate for a given UAV trajectory and the transmit power. Whereas, in \cite{fan14}, the secrecy rate was maximized by jointly optimizing IRS phase shift, UAV trajectory, and the transmit power. In \cite{xie15}, the beamforming, power allocation, and the IRS phase shift were jointly optimized to minimize the transmit power. The total transmit power from the access point (AP) was minimized in \cite{sam16} through joint optimization of AP transmit beamforming and IRS reflect beamforming under the SNR constraints. A dynamic passive beamforming scheme was proposed in \cite{yif17}, where IRS reflection coefficients were adjusted dynamically over different time slots to achieve higher passive beamforming gain. An IRS was applied to generalized frequency division multiplexing (GFDM) communication in \cite{tabat34} to enhance system performance. In multi-IRS aided two-way full-duplex communication systems, joint optimization of IRS location and size was carried out in \cite{chris35}. Optimal positions for different IRS with fixed and variable heights were investigated in \cite{zhu36}. The system sum rate was maximized in \cite{yu18} by jointly optimizing the source precoders and the IRS phase shift matrix in a full-duplex MIMO two-way communication. Alternatively, \cite{you19} maximized sum rate by optimizing the IRS location and showed that it should be near either the AP or the user. This work is used as a benchmark for comparing the proposed work. \emph{However, from the state-of-the-art, we need to explore the joint optimization of the phase shifts as well as the deployment of IRS that includes the important factors such as free space loss, reflection loss factor, physical dimension of the IRS panel, and the radiation pattern of the transmit signal.}

\subsection{Motivation and Key Contributions}
 Using IRS, the communication channel can be intelligently engineered, leading to enhanced end-to-end communications and effective co-channel interference suppression. This paper presents a communication model and derives a reliable expression for received power in an IRS-assisted wireless communication system, emphasizing the significant impact of reflective phase shift and precise IRS location. To enhance its performance, in this work, we analyze the joint optimizations of the underlying variables. The key contribution of this work is three-fold. (1) Leveraging the available radar communication data, we derive the expression for the reflection coefficient of the IRS panel. Along with it, in the two scenarios: a single element and multi-elements IRS, we derive the expressions for received power at a user that include the factors such as free space loss, reflection loss factor, physical dimensions of the IRS panel, and the radiation pattern of the transmit signal. (2) We introduce a joint optimization approach and semi-adaptive schemes for maximizing received power. In location optimization for a fixed phase shift, we exploit approximate method using polynomial curve fitting method for an optimal solution with reduced complexity. The optimization problem's convexity and global optimality are proven, yielding a closed-form solution. This optimal solution reveals insights into constructive and destructive interference when partially optimizing phase shifts and jointly optimizing underlying variables. (3) Through numerical analysis, we explore received power variations based on phase shifts and IRS location, yielding crucial insights into the optimal solution. We also examine phase shift variation with location across different wavenumber values. Additionally, we analyze how system parameters such as the reflection coefficient and the count of IRS elements influence received power. Finally, we rigorously compare the performance of our proposed scheme against a benchmark. 

\section{Problem Definition}\label{Sec:Sys_Model} 
In this section, first, we derive the expression of reflection coefficient of IRS by leveraging the existing data of radar communications. Then, we establish the transmitted power-received power relationship initially considering a single IRS element by introducing a two ray model.
Initially focusing on a single IRS element with an area of s $m^2$ and a single phase shift, we gain insights into jointly optimizing the IRS location and reflection phase shift to maximize received power, which directly correlates with signal-to-noise ratio (SNR) and throughput improvement in the single-user scenario. Consequently, the derived optimal phase shift and IRS location will maximize throughput in the IRS-assisted communication. Next, we extend our analysis to a multi-ray model, incorporating multiple IRS elements, and derive the corresponding received power expression at the receiver. Joint optimization problem for both the cases is formulated in order to enhance the received power keeping phase shift and location of the IRS as underlying variables. 

\subsection{IRS Reflection Model}
Here, we consider an IRS panel consisting of a single element that reflects the incoming signal. During reflection, a part of the signal energy is absorbed and the reflection is not necessarily distributed evenly in all the required possible directions. Moreover, it is difficult to calculate the reflection coefficient of an IRS as it cannot receive the signal on its own. To realize it, we examine the existing received radar data obtained after reflection from different objects as follows.
\begin{table}[!t]
	\centering  
	\caption{Ability of radar to detect aircraft of different sizes~\cite{rez01}}  
	\centering
	\label{tab:table1}  \small
	\scalebox{0.85}{\begin{tabular}{|l|c|c|c|c|}
			\hline
			Object (Aircraft)&Length ($m$) &Width ($m$) & Area ($m^2$) & RCS ($m^2$) \\
			\hline
			Large Fighter                 & 16        & 6        & 96       & 6       \\
			\hline
			Four Passenger  & 13.2      & 5.2      & 68.64    & 2       \\
			Jet/Small Fighter&&&&\\
			\hline
			Small Single  & 10.49     & 3.8      & 39.862   & 1       \\
			Engine Aircraft &&&&\\
			\hline
			Conventional    & 6.32      & 3.2      & 20.224   & 0.5\\
			Winged Missile &&&&\\
			\hline	
	\end{tabular}}	
\end{table}
\begin{figure}[!t] 
	\centering  \includegraphics[width=3.4in]{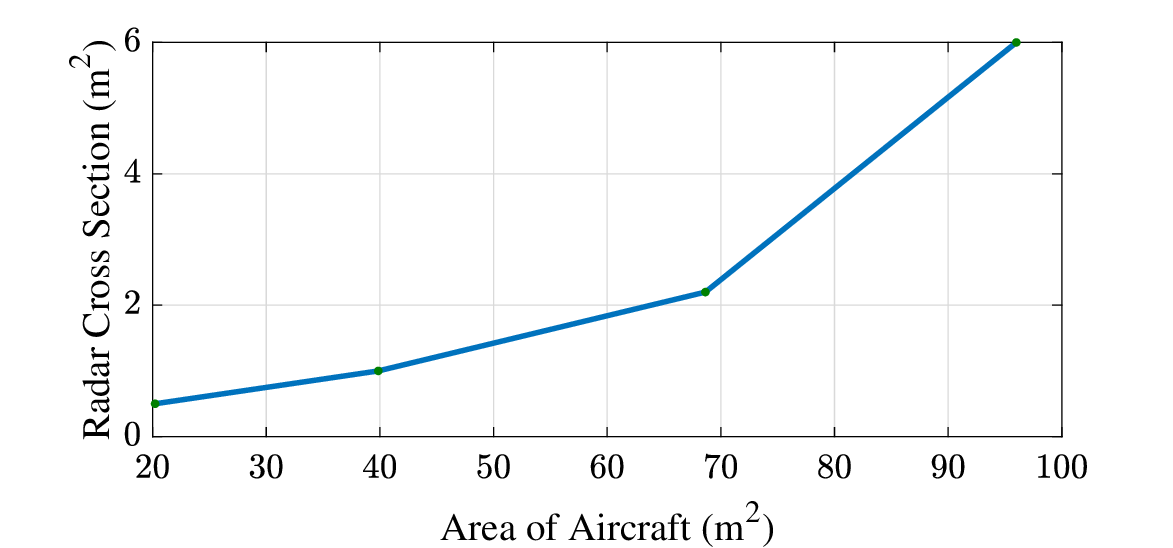}
	\caption{\small From Table~\ref{tab:table1}, obtained variation of RCS, $\sigma$ with area $s$.}    
	\label{fig:sig_s_rel}\vspace{0mm}
\end{figure}
 Indeed, the electromagnetic waveform emitted by the radar is spherical in shape. A part of the waveform is scattered and absorbed on the surface of the aircraft and the remaining is reflected back to the radar to detect it. Specifically, the measurement of radar cross section (RCS), $\sigma$ using reflected signal determines the detection of the distant object. From~\cite{rez05}, omitting the shape and material of the objects, the size of the aeroplanes used in world war~II and corresponding measured RCSs are listed in Table~\ref{tab:table1} and the plot for RCS, $\sigma$ with area of the aircraft, $s$ is drawn in Fig.~\ref{fig:sig_s_rel}. Here, the area of an aircraft is estimated by its measurement of length of the wing and length of the fuselage. In Fig.~\ref{fig:sig_s_rel}, $\sigma$ increases with $s$, to model the relationship mathematically, using Image fitting tool, $\sigma$ in $s$ can be expressed as:
\begin{align}\label{eq:sig_s_rel}\small
	\textstyle\sigma=\textstyle 0.2 s^{2}-0.0961 s-0.76.
\end{align}
 In~\eqref{eq:sig_s_rel}, it is observed that $\sigma$ is approximately proportional to the square of $s$. As $\sigma$ is related to reflected power of the scattered signal from the target object to the radar~\cite{kar02}, the reflection coefficient $\Gamma$ of the object is given by:
\begin{align}\label{eq:ref_coeff}\small
	\textstyle\Gamma=\textstyle\sqrt{\sigma}=\sqrt{0.2 s^{2}-0.0961 s-0.76}
\end{align}
As the obtained expression of $\Gamma$ in~\eqref{eq:ref_coeff} relies on the size of the object, so, it is not limited to IRS, but true for any physical object of interest of same area $s$. 

\subsection{Transmission Model for Received Power at User $\mathcal{U}$ for Single IRS Element}
From the conventional Friis transmission model~\cite{kum03}, the received power, $P_r$ at a user for the given transmit power, $P_t$ is given by
\begin{align}\label{eq:friis}\small
	\textstyle\frac{P_{r}}{P_{t}}=\textstyle G_{t} G_{r}\left(1-\left|\Gamma_{t}\right|^{2}\right)\left(1-\left|\Gamma_{r}\right|^{2}\right)\left(\frac{\lambda}{4 \pi l}\right)^{2}\left|\hat{e_{n}} \cdot \hat{e}_{r}\right|^{2},
\end{align} 
where $l$ is the distance between transmitter and receiver, $\lambda$ is the wavelength of transmit electromagnetic wave, and $|\hat{e_{n}} \cdot \hat{e}_{r}|^{2}$ is the polarization loss factor. $G_t$ and $G_r$ are the gains, whereas $\Gamma_t$ and $\Gamma_r$ are the reflection coefficients at the transmit and receive antennas, respectively. However, as described in~\cite{kum03}, the Friis equation in~\eqref{eq:friis} is effectively applicable for free space propagation. In practice, IRS-assisted wireless nodes are deployed near ground, the reflected signal components are also received at the receiver apart from the signal through direct path. Therefore, we need to develop a communication model for the IRS-assisted system as shown in Fig.~\ref{fig:sys_topo}. 
\begin{figure}[!t] 
	\centering  \includegraphics[width=3.0in]{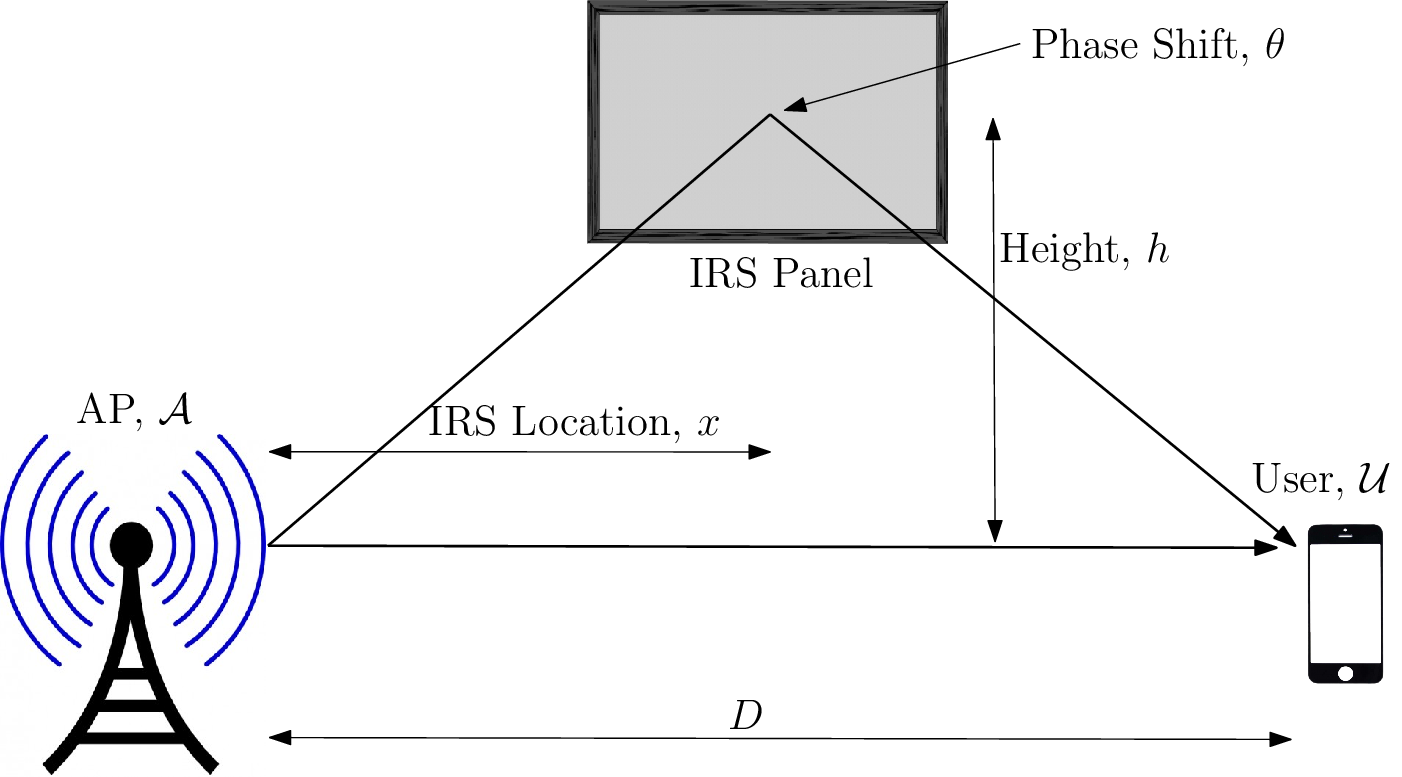}
	\caption{\small IRS-assisted two-ray communication model.}    
	\label{fig:sys_topo}\vspace{-0mm}
\end{figure}
The two-ray reflection model for the network has been shown in Fig.~\ref{fig:sys_topo} where the transmit signal from the AP, $\mathcal{A}$ is received to user, $\mathcal{U}$ through reflection from the IRS panel and direct path. Note that the transmit antenna is omnidirectional in nature. The distance between $\mathcal{A}$ and $\mathcal{U}$ is $D$ and IRS panel is offset by height, $h$ from the direct path. We assume that the panel remain at the constant height, $h$, but can adjust the location, $x$ along the direct path between $\mathcal{A}$ and $\mathcal{U}$. The reflected and direct received signal components add constructively or destructively depending on the phase shift via the two paths. For the wavenumber $k=2\pi/\lambda$, the phase shift in direct path is $kD$, whereas the reflected component through IRS panel faces the phase shift $kd$ due to path length $d$ and $\theta$ because of reflection at the panel.  Therefore, for the negligible reflection coefficients at the transmitter and receiver, the received power, $P_r$ at $\mathcal{U}$ for a given transmit power $P_t$ is:
\begin{align}\label{eq:rx_pow}
	\textstyle P_{r}&=\textstyle {P_{t}}\times\left(\frac{\lambda}{4 \pi}\right)^{2}\left|\frac{1}{D} e^{ (-j k \mathrm{D})}+\frac{\Gamma}{d} e^{ \left(j \theta\right)}e^{\left(-j k d\right)}\right|^{2},\\\label{eq:d_exp}
	\textstyle d&=\textstyle\sqrt{x^{2}+\mathrm{h}^{2}}+\sqrt{\left(D-x\right)^{2}+\mathrm{h}^{2}}
\end{align}
where $d$ is the path length of the reflected signal and $\Gamma$ is the reflection coefficient loss as given by~\eqref{eq:ref_coeff}. Note that the phase shift $\theta$ at the IRS panel is tractable and can be adjusted using its control unit.
\begin{figure}[!t] 
	\centering  \includegraphics[width=3.4in]{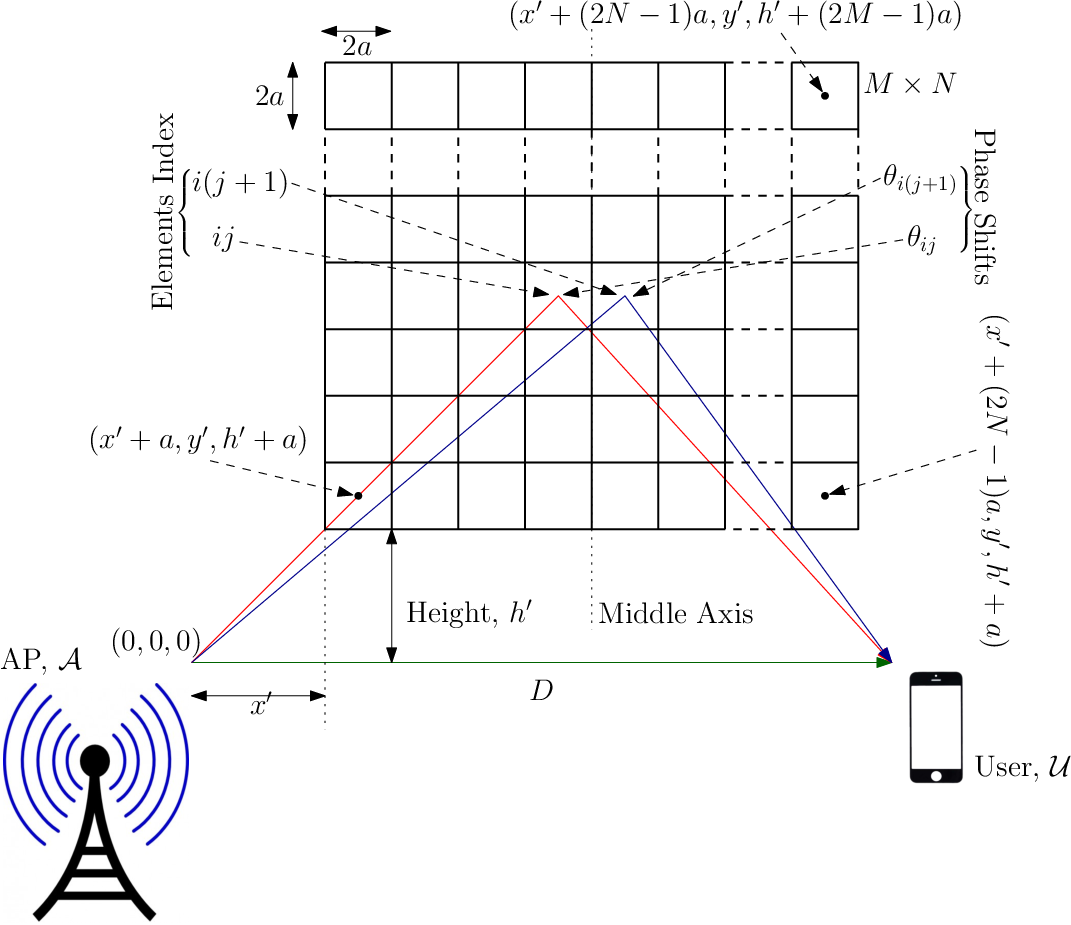}
	\caption{\small IRS-assisted multi-ray communication model.}    
	\label{fig:sys_topo_multi}\vspace{-2mm}
\end{figure}

\subsection{Transmission Model for Received Power at User $\mathcal{U}$ for Multiple IRS Elements}
Now, let us consider the case of multi IRS elements IRS as shown in Fig.~\ref{fig:sys_topo_multi}. Here, in the multi-element IRS-assisted communication as shown in Fig.~\ref{fig:sys_topo_multi}, the downlink transmission from the AP, $\mathcal{A}$ to user, $\mathcal{U}$ takes place through direct as well as the IRS-reflected links. The IRS panel is assumed to be of $M\times N$ elements (size of each square is of length $2a$) represented by the index $\{ij|0\le i \le M-1, 0\le j \le N-1\}$. The phase shift after reflection at $ij$th element is $\theta_{ij}$. Reflection coefficient, $\Gamma$ is considered to be the same at each element due to symmetric characteristics. In order to represent the elements in a three-dimensional Cartesian coordinate system, we consider the origin at the AP, $\mathcal{A}$ as shown in Fig.~\ref{fig:sys_topo_multi}. In general, the coordinate of the center of the $ij$th element is $(x'+(2j+1)a, y',  h'+(2i+1)a)$, where $x'$ is the distance of the leftmost point of the IRS panel from $\mathcal{A}$, $y'$ is distance away from the direct path between $\mathcal{A}$ and $\mathcal{U}$, and $h'$ is the height of the bottom-most point of the panel from the direct path. Note that without loss of generality, in Fig.~\ref{fig:sys_topo_multi}, the distance $y'$ is not shown for the convenience of the descriptive representation.

Using the similar procedure to obtain \eqref{eq:friis}, for the multi-ray communication model in Fig.~\ref{fig:sys_topo}, the received power, $\widetilde{P}_r$ at user $\mathcal{U}$ for given transmit power, $P_t$ from the AP, $\mathcal{A}$ is obtained as
\begin{align}\label{eq:rx_pow_multi}
	\widetilde{P}_r={P_{t}}\left(\frac{\lambda}{4 \pi}\right)^{2}\left|\frac{1}{D} e^{ (-j k \mathrm{D})}+
	\sum_{i=0}^{M-1}\sum_{j=0}^{N-1}\frac{\Gamma}{d_{ij}} e^{ \left(j \theta_{ij}\right)}e^{\left(-j k d_{ij}\right)}\right|^{2},
\end{align} 
where we assume that the end nodes reflection coefficients, $\Gamma_{t}=\Gamma_{r}=0$ and the gains, $G_t=G_r=1$. $\Gamma$ is the constant reflection coefficient loss at each element.  Moreover, the path length, $d_{ij}$ of the reflected signal through $ij$th element can be expressed as
\begin{align}
	\label{eq:d_exp_multi}\nonumber
	d_{ij}=&\sqrt{{(x'+(2j+1)a)}^{2}+y'^2+{(h'+(2i+1)a)}^{2}} \\
	&+\sqrt{{(D-x'-(2j+1)a)}^{2}+y'^2+{(h'+(2i+1)a)}^{2}},
\end{align}
for $i\in\{0,1,\cdots,M-1\}$ and $j\in\{0,1,\cdots,N-1\}$. Note that in~\eqref{eq:rx_pow_multi}, the phase shift due to direct path is $kD$ whereas the reflected component through $ij$th element of the IRS panel encounters a phase shift $kd_{ij}$ and $\theta_{ij}$ due to the path length $d_{ij}$ and the reflection, respectively. 

\subsection{Problem Formulation}\label{sec:prob_form}

\subsubsection{Joint Optimization Problem for Single IRS Element}

In case of reflection from a single IRS element as shown in Fig.~\ref{fig:sys_topo}, to enhance the constructive interference of the two received rays at the user, $\mathcal{U}$, the IRS panel location, $x$ and reflection phase shift, $\theta$ at the panel can be regulated. Therefore, the received power, $P_r$ in~\eqref{eq:rx_pow} under the constraints of underlying variables can be expressed as:
\begin{eqnarray}
	\begin{aligned}
		\text{(P0):}&\hspace{1mm} \textstyle\underset{x,\;\theta}{\text{maximize}}
		\quad P_r \\
		& \hspace{1mm} \text{subject to}\quad C1\!: \textstyle 0\le x \le D,\;\;\;  C2\!: \textstyle 0\le \theta \le 2\pi, \nonumber
	\end{aligned}
\end{eqnarray}
where the constraints $C1$ and $C2$ provide the boundary conditions on the variables $x$ and $\theta$. As objective function $P_r$ is nonconvex in $x$ and $\theta$, the joint optimization problem (P0) is nonconvex. 

In order to realize the solution for the problem (P0), objective function, $P_r$ can be further simplified  using Lemma~\ref{lemma1}.
\begin{lemma}\label{lemma1}
	The received power, $P_r$ in~\eqref{eq:rx_pow} can be expressed in more tractable form as:
	\begin{align}\label{eq:Pr_trac}
		\textstyle P_r=\textstyle {P_{t}}\left( \frac{\lambda}{4\pi} \right) ^2\left( \frac{1}{D^2}+\frac{\Gamma ^2}{d^2}+\frac{2\Gamma}{dD}\cos \left( kd-\theta -kD \right) \,\, \right)
	\end{align}
\end{lemma}
\begin{proof}
	The proof is discoursed in Appendix~\ref{App_A}.
\end{proof}

\subsubsection{Joint Optimization Problem for Multiple IRS Elements}

As shown in Fig.~\ref{fig:sys_topo_multi}, let us consider for the case of multiple IRS elements. In order to maximize $\widetilde{P}_r$ in~\eqref{eq:rx_pow_multi} by jointly optimizing the phase shifts, $\{\theta_{ij}\}$ and IRS location, $x'$, we need to express the equation in a more tractable form using Lemma~\ref{lemma2} as described below.

\begin{lemma}\label{lemma2}
	The received power, $\widetilde{P}_r$ in~\eqref{eq:rx_pow_multi} can be further simplified into the expression as shown in~\eqref{eq:Pr_trac_multi} at the top of next page.
\end{lemma}

From~\eqref{eq:Pr_trac_multi}, it is to be noted that the amount of the received power, $\widetilde{P}_r$ highly depends on the cosine factors which leads to the constructive and destructive interference at user, $\mathcal{U}$. As these factors are function of path lengths, $\{d_{ij}\}$ and phase shifts, $\{\theta_{ij}\}$, therefore these factors need to be optimized in order to enhance $\widetilde{P}_r$. Note that from~\eqref{eq:d_exp_multi}, $\{d_{ij}\}$ depend on $x'$, therefore, the corresponding optimization problem can be formulated as
\begin{eqnarray}
	\begin{aligned}
		\text{(P1):}&\hspace{1mm} \underset{x',\;\{\theta_{ij}\}}{\text{maximize}}
		\quad \widetilde{P}_r \\
		& \hspace{1mm} \text{subject to}\quad  \widetilde{C}1\!:  0\le x' \le D,\;\;\;  \widetilde{C}2\!: 0\le \theta_{ij} \le 2\pi, \nonumber
	\end{aligned}
\end{eqnarray}
where $\widetilde{C}1$ and $\widetilde{C}2$ are the given bounds for the location and phase shifts, respectively. From the investigation of Hessian matrix of $\widetilde{P}_r$ in $x'$ and $\{\theta_{ij}\}$, we find that $\widetilde{P}_r$ is jointly nonconvex which implies that the problem, (P1) is nonconvex. Next, in the solution methodology section, we obtain the sub-optimal solution for (P1) and its global optimality is realized using the obtained numerical results in the numerical section.
\begin{figure*}[!t]
\begin{equation}\label{eq:Pr_trac_multi}
	\widetilde{P}_r\hspace{-1mm}={P_{t}}\left(\frac{\lambda}{4\pi}\right) ^2\hspace{-1mm}\left(\hspace{-1mm} \frac{1}{D^2}\hspace{-0.5mm}+\hspace{-1mm}\sum_{i,j}\frac{\Gamma ^2}{d_{ij}^2}+\hspace{-1mm}\sum_{i,j}\frac{2\Gamma}{Dd_{ij}}\cos \left( kd_{ij}\hspace{-1mm}-\hspace{-1mm}\theta_{ij}\hspace{-1mm}-\hspace{-1mm}kD \right) +\hspace{-1mm}\left.\sum_{i,j}\sum_{\substack{m,n\\\sim(m=i,n=j)}}\frac{2\Gamma^2}{d_{ij}d_{mn}}\right.\cos \left( \theta_{ij}\hspace{-1mm} -kd_{ij}\hspace{-1mm} - \theta_{mn}\hspace{-1mm}  + kd_{mn} \right)\hspace{-1mm} \right).
\end{equation}
	\hrulefill
\end{figure*}
\begin{proof}
	See Appendix~\ref{App_B}.
\end{proof}

\section{Proposed Solution Methodology}\label{sec:sol_methodology}

As described in Section~\ref{sec:prob_form}, it is observed that both the problems (P0) and (P1) are nonconvex. Here, in this section, we first consider the case of single IRS element where we obtain the jointly optimal solution for problem (P0) followed by solution methodology for semi-adaptive schemes to optimize the individual variable while keeping other fixed. Next, we consider the case of multi-element IRS, where we realize its jointly optimal solution by first optimizing the phase shifts at the IRS elements followed by the optimization of IRS location at the obtained optimal phase shifts.

\subsection{Jointly Optimal IRS Location and Phase Shift for Single Element IRS}\label{sec:joint_loc_phase}

As the problem (P0) is jointly nonconvex, to determine its feasible optimal solution, we investigate the variation of the objective function, $P_r$ with reflected path length, $d$ and phase shift, $\theta$. Note that one of the underlying variable of the problem (P0) is location, $x$ instead of $d$, but to reduce complexity, initially, we find the optimal path length, $d^*$, then, using~\eqref{eq:d_exp}, the optimal location, $x^*$ is determined.  If we examine the expression of $P_r$ in~\eqref{eq:Pr_trac}, $\theta$ is inside the cosine factor. Therefore, to enhance $P_r$, $\theta$ can be set as a function of $d$ such that $(kd-\theta-kD)=2m\pi$ which gives $\cos(kd-\theta-kD)=1$, where $m\in\mathbb{Z}$ is an integer. Now, the remaining term in~\eqref{eq:Pr_trac} is represented as in~\eqref{eq:exp_Pr_hat}, is a function of $d$ and its first and second derivatives are expressed below in~\eqref{eq:first_Pr_hat} and \eqref{eq:second_Pr_hat}.
\begin{subequations}\label{eq:derivatives}
	\begin{align}\label{eq:exp_Pr_hat}
		\textstyle \widehat{P}_r&=\textstyle {P_{t}}(\lambda/4\pi) ^2\left( \frac{1}{D^2}+\frac{\Gamma ^2}{d^2}+\frac{2\Gamma}{dD}\right)\\\label{eq:first_Pr_hat}
		\textstyle \frac{\partial \widehat{P}_r}{\partial d}&=\textstyle {P_{t}}\left( \frac{\lambda}{4\pi} \right) ^2\left(-2\frac{\Gamma ^2}{d^3}-2\frac{\Gamma }{Dd^2}\right)<0\\\label{eq:second_Pr_hat}
		\textstyle \frac{\partial^2\widehat{P}_r}{\partial d^2}&=\textstyle {P_{t}}\left( \frac{\lambda}{4\pi} \right) ^2\left(6\frac{\Gamma  ^2}{d^4}+4\frac{\Gamma}{Dd^4}\right)>0.
	\end{align} 
\end{subequations}
From~\eqref{eq:derivatives}, the first derivative in $d$ is always negative, whereas the second derivative is always positive. Thus, $\widehat{P}_r$ is strictly convex in $d$, but it always decreases with $d$. Therefore, $\widehat{P}_r$ achieves its maximum value at the minimum value of $d$. As $d$ in~\eqref{eq:d_exp} is a function of $x$, its minimum value, $d_{\min}$ is obtained as: $\frac{\partial d}{\partial x}=\frac{x}{\sqrt{h^{2}+x^{2}}}-\frac{D-x}{\sqrt{(D-x)^{2}+h^{2}}}=0$ which gives optimal IRS location, $x^*=D/2$, after substituting it in $d$, we get $d_{\min}=2\sqrt{h^2+\frac{D^2}{4}}$. To determine the optimal phase shift, $\theta^*$, we substitute, $d_{\min}$ in $(kd-\theta-kD)=2m\pi$, it gives $\theta^*=k\Big( 2\sqrt{h^2+\frac{D^2}{4}}-D\Big)+2m\pi$. But, under the constraint of $C2$, we take $m=0$ which gives $\theta^*$ as represented below in~\eqref{eq:joint_opt_theta}. Thus, the jointly optimal solution, $(x^*,\theta^*)$ can be expressed as: 
\begin{subequations}
	\begin{align}\label{eq:joint_opt_x}
		\textstyle x^*&=\textstyle \frac{D}{2},\\\label{eq:joint_opt_theta}
		\textstyle \theta^*&=\textstyle k\left( 2\sqrt{h^2+\frac{D^2}{4}}-D\right).
	\end{align}
\end{subequations}
Next, using the semi-adaptive schemes for individual optimization of $x$ and $\theta$ as described below, we get more insights on the global optimality of the obtained optimal solution.

\subsection{The Optimal Location of Single-Element IRS}\label{sec:opt_loc}

In this semi-adaptive scheme, we optimize the location, $x$ of IRS panel at a fixed phase shift, $\theta$ to maximize, $P_r$ and its optimization formulation can be expressed as:
\begin{eqnarray}\label{eq:prob_P2}\nonumber
	\begin{aligned}
		\text{(P2):}&\hspace{1mm}\underset{x}{\text{maximize}} \quad P_{r}\\
		&\hspace{1mm}\text {subject to: }\hspace{1.5mm} C1\!: 0 \leqslant x \leqslant D
	\end{aligned}
\end{eqnarray}
Again, using the same procedure as in Section~\ref{sec:joint_loc_phase}, to reduce the complexity, first we determine the optimal path length, $d^*$, thereafter, using~\eqref{eq:d_exp}, we determine the optimal location, $x^*$. As the problem (P2) is nonconvex, to realize its optimal solution, we find the roots of $\frac{\partial P_r}{\partial d}=0$ as given by
\begin{align}\label{eq:Pr_der_d}
	\textstyle \frac{\partial P_r}{\partial d}=\textstyle \left( \frac{\lambda}{4\pi} \right) ^2\frac{2\Gamma}{dD}\sin \left( kd-\theta -kD \right)=0.
\end{align}
As~\eqref{eq:Pr_der_d} has a trigonometric factor, $\sin(kd-\theta -kD)$ which gives infinite many roots, therefore, to realize the roots in more tractable form and to confine the feasible roots in finite range, we use polynomial curve fitting for the received power, $P_r$ in~\eqref{eq:Pr_trac} under given range of underlying variable. In the expression of $P_r$, the factor, $\cos(kd-\theta -kD)$ needs to be transformed using polynomial curve fitting to completely represent $P_r$ in polynomial form. For the domain $z\in[-\pi/2, \pi/2]$, the polynomial expression of $\cos(z)$ using curve fitting is:
\begin{align}\label{eq:cos_poly}
	\textstyle \cos \left( z \right) &\approx \textstyle 0.0259z^4-0.4507z^2+0.9772.
\end{align} 
Substituting $z=kd-\theta -kD$ in~\eqref{eq:cos_poly}, $\cos(kd-\theta -kD)$ can be expressed as:
\begin{align}\label{eq:cos_poly_act}
	\textstyle \cos(kd-\theta-kD)=\textstyle Md^4+Nd^3+Pd^2+Qd+R,
\end{align}
where $M$, $N$, $P$, $Q$, and $R$ are expressed below.
\begin{align}\nonumber
	\textstyle M &=\textstyle 0.0259k^4\\\nonumber
	\textstyle N &=\textstyle 0.0259(-2k^3\theta-2k^4\theta-2k^3\theta-2k^4D)\\\nonumber
	\textstyle P &=\textstyle 0.0259(k^2\theta^2+k^4D^2+2\theta k^3D+4k^2\theta^2+8k^3\theta D\\\nonumber
	& \textstyle +4k^4D^2+\theta^2k^2+k^4D^2-2k^3\theta D^2 + 2 \theta k^3-0.4507k^2\\\nonumber
	\textstyle Q &=\textstyle 0.0259(-2k\theta^3 -2k^3\theta D^2 -4k^2\theta^2 D -2k^2D\theta^2\\\nonumber
	&\textstyle -2k^4D^3-4\theta k^3D^2 -2k \theta^3 -2k^2 \theta^2 D -2k^4D^3\\\nonumber
	&\textstyle -4 k^2\theta^2 D-4\theta k^3D^2)-0.4507(-2k\theta -2k^2D)\\\nonumber
	\textstyle R &=\textstyle 0.0259(\theta^4 + \theta^2k^2D^2 +2\theta^3kD +k^2D^2 \theta^2+k^4D^4 \\\nonumber
	& \textstyle + 2\theta k^3D^3 + 2\theta^3 kD + 2 \theta k^3 D^3+ 4 \theta^2 k^2 D^2) + 0.9772
\end{align}
Note that for given $d$, $\theta$, $k$, and $D$, if $z<-\pi/2$ or $z>\pi/2$, then we perform $z+2m\pi$ or $z-2m\pi$ respectively to fetch it in the range, $[-\pi/2,\pi/2]$. After substitution of cosine factor in~\eqref{eq:cos_poly_act} to~\eqref{eq:Pr_trac}, the received power represented as $\overline{P}_r$, the optimal solution of (P1) can be realized using equation $\frac{\partial \overline{P}_r}{\partial d}=0$ as given by
\begin{align}\label{eq:Pr_poly}
	\textstyle \frac{\partial \overline{P}_r}{\partial d}=\textstyle U d^5+ V d^4 + W d^3 + X d + Y=0
\end{align} 
where $U$, $V$, $W$, $X$, and $Y$ are expressed below.
\begin{align}\nonumber
	\textstyle U = \textstyle\frac{6 \Gamma M}{D},
	V =\frac{4 \Gamma N}{D},
	W =\frac{2 \Gamma P}{D},
	X =\frac{-2 \Gamma R}{D},
	Y =-2 \Gamma^2
\end{align}
In contrast,~\eqref{eq:Pr_poly} has finite (five) roots which can be realized readily and some of the roots are eliminated as they must be real and should satisfy $d>0$. In case of more than one feasible solution, the one is set to globally optimal solution, $d^*$ at which $\overline{P}_r$ achieves its maximum value. Further, from~\eqref{eq:d_exp}, for obtained $d^*$, two values of optimal location, $x^*$ denoted as $\{x_1,x_2\}$ are obtained which are symmetric about $x=D/2$. As both gives the same maximum value of $P_r$, one of them is set as optimal location as expressed as: $x^*=\{x_1 \text{ or } x_2|d=d^*\}$. Moreover, the discussion is illustrated in the numerical section using the obtained results. 

\subsection{The Optimal Phase Shift for Single-Element IRS}\label{sec:opt_phase}
In this semi-adaptive scheme, the received power, $P_r$ is maximized by optimizing the phase shift, $\theta$ at the IRS-panel while keeping its location, $x$ at a fixed value. The optimization problem for the scheme is given by:
\begin{eqnarray}\label{eq:prob_P3}\nonumber
	\begin{aligned}
		\text{(P3):}&\hspace{1mm}\underset{\theta}{\text{maximize}} \quad P_{r}\\
		&\hspace{1mm}\text {subject to: }\hspace{1.5mm} C2\!: 0<\theta<2\pi
	\end{aligned}
\end{eqnarray}
Again, the problem (P3) is nonconvex, but it is convex in the restricted range of, $\theta$.  To realize its feasible globally optimal solution, its first derivative, $\frac{\partial P_r}{\partial \theta}=0$ and second derivative, $\frac{\partial^2 P_r}{\partial \theta^2}$ can be investigated as below.
\begin{align}\label{eq:Pr_first_der_theta}
	\textstyle \frac{\partial P_r}{\partial \theta} &= \textstyle \left( \frac{\lambda}{4\pi} \right) ^2\frac{2\Gamma}{dD}\sin \left( kd-\theta -kD \right)=0\\\label{eq:Pr_sencond_der_theta}
	\textstyle \frac{\partial^2 P_r}{\partial \theta^2} &=\textstyle -\left( \frac{\lambda}{4\pi} \right) ^2\frac{2\Gamma}{dD}\cos \left( kd-\theta -kD \right)
\end{align}
From~\eqref{eq:Pr_sencond_der_theta}, $\frac{\partial^2 P_r}{\partial \theta^2}<0$ for $(kd-\theta -kD)\in\big((4m-1)\frac{\pi}{2}+k(d-D), (4m+1)\frac{\pi}{2}+k(d-D)\big)$; $m\in\mathbb{Z}$. Therefore, in the given domain, $P_r$ is concave in $\theta$ which implies (P3) is convex and using~\eqref{eq:Pr_first_der_theta}, globally optimal solution is: 
\begin{align}
	\theta^*=k(d-D).
\end{align}
\subsection{Optimal Phase Shifts for Multi-Elements IRS}
The received power $\widetilde{P}_r$ in~\eqref{eq:rx_pow_multi} can be enhanced by choosing the phase shifts, $\{\theta_{ij}\}$ for a given IRS location, $x'$ such that the cosine factors achieve its maximum value $1$ which gives the following relationship.
\begin{subequations}
	\begin{align}\label{eq:theta_a}
		&(kd_{ij}-\theta_{ij}-kD) = 2p\pi,\\\label{eq:theta_b}
		&(\theta_{ij}-kd_{ij}-\theta_{mn} + kd_{mn}) =2q\pi,
	\end{align}
\end{subequations} 
where $p,q\in\mathbb{Z}$ are integers. From~\eqref{eq:theta_a}, the phase shift can be calculated as, $\theta_{ij}=kd_{ij}-kD-2p\pi$. On further substituting it in~\eqref{eq:theta_b}, we get, $\theta_{mn}=kd_{mn}-kD-2(p+q)\pi$. As we observe that phase shifts, $\theta_{ij}$ and $\theta_{mn}$ are in the same form, therefore, under the given constraint $\widetilde{C}1$ in (P1) the optimal phase shift, $\tilde{\theta}_{ij}$ is given by
\begin{align}\label{eq:opt_phase}
	\tilde{\theta}_{ij}=kd_{ij}-kD+2z\pi\;(\text{mod}\; 2\pi),
\end{align}
where $z\in\mathbb{Z}$. Now, after substituting $\{\tilde{\theta}_{ij}\}$ in~\eqref{eq:Pr_trac_multi}, the received power, $\widehat{\widetilde{P}}_r$ at the optimal phase shift can be expressed as
\begin{align}\label{eq:rx_power_opt_multi}
	\widehat{\widetilde{P}}_r\hspace{-1mm}=&{P_{t}}\left(\hspace{-1mm} \frac{\lambda}{4\pi}\hspace{-1mm} \right)^2\hspace{-1mm}\left(\hspace{-1mm} \frac{1}{D^2}\hspace{-1mm}+\hspace{-1mm}\sum_{i,j}\frac{\Gamma ^2}{d_{ij}^2}\hspace{-1mm}+\hspace{-1mm}\sum_{i,j}\hspace{-0.5mm}\frac{2\Gamma}{d_{ij}D}\hspace{-1mm}+\hspace{-1mm} \sum_{i,j}\hspace{-3mm}\sum_{\substack{m,n\\\sim(m=i,n=j)}}\hspace{-3mm}\frac{2\Gamma^2}{d_{ij}d_{mn}} \hspace{-1mm} \right).
\end{align}
Next, using~\eqref{eq:rx_power_opt_multi}, we find the optimal location for multiple IRS elements at the obtained optimal phase shifts.

\subsection{Optimal IRS Location for Multi-Elements IRS}
To determine the globally optimal IRS location, $x'^*$ at the optimal phase shifts, $\{\tilde{\theta}_{ij}\}$, here we investigate the variation of $\widehat{\widetilde{P}}_r$ in path lengths, $\{d_{ij}\}$ using Lemma~\ref{lemma3}.
\begin{lemma}\label{lemma3}
	The received power, $\widehat{\widetilde{P}}_r$ in~\eqref{eq:rx_power_opt_multi} is jointly unimodal in $\{d_{ij}\}$ and achieves its globally maximum value at jointly minimum value of $\{d_{ij}\}$.
\end{lemma}
\begin{proof}
	If we consider the gradient of $\widehat{\widetilde{P}}_r$ in the reflected path lengths, $\{d_{ij}\}$, then, we have
	\begin{align}\label{eq:grad_d}
		\nabla_d \widehat{\widetilde{P}}_r=\left[-\frac{2}{d_{ij}^3}\left(\Gamma+\frac{d_{ij}}{D}+\sum_{\substack{m,n\\\sim(m=i,n=j)}}\frac{\Gamma d_{ij}}{d_{mn}}\right)\right]_{i=0,j=0}^{M-1,N-1}\hspace{-5mm}\prec\mathbf{0},
	\end{align}
	where $\nabla_d$ is the gradient of the function in path lengths, $\{d_{ij}\}$. Here, it can be noticed that the elements in the vector of length $M\hspace{0.1mm}N$ in~\eqref{eq:grad_d} are $<0$. It implies that $\widehat{\widetilde{P}}_r$ is jointly unimodal and strictly decreasing in $\{d_{ij}\}$. Thus, $\widehat{\widetilde{P}}_r$ gets its globally maximum value at jointly minimum value of $\{d_{ij}\}$.
\end{proof}	
The jointly minimum value of $\{d_{ij}\}$ in location, $x'$ can be realized by minimizing the maximum value of $\{d_{ij}\}$ under the constraint of the underlying variable $x'$ which is formulated as	
\begin{eqnarray}
	\begin{aligned}
		\text{(P4):}&\hspace{1mm} \underset{x'}{\text{min}}\;\underset{i,j}{\text{max}}
		\;\; \{d_{ij}\} \\
		& \hspace{1mm} \text{subject to}\quad \widetilde{C}1 \nonumber.
	\end{aligned}
\end{eqnarray}
The minimax problem (P4) can be equivalently expressed as
\begin{eqnarray}
	\begin{aligned}
		\text{(P5):}&\hspace{1mm} \underset{x',t}{\text{min}}\;\; t \\
		& \hspace{1mm} \text{subject to}\quad \widetilde{C}1, \widetilde{C}3: t\ge d_{ij}\; \forall i,j \nonumber,
	\end{aligned}
\end{eqnarray}
where the constraint, $\widetilde{C}3$ ensures that $t$ is always greater than the path lengths, $\{d_{ij}\}$. In order to realize its optimal solution, its Lagrangian function $\mathcal{L}$ is given by
\begin{align}\label{eq:lagrangian}
	\mathcal{L}(t,x',\lambda_{ij})=t+\sum_{i,j}\lambda_{ij}(d_{ij}-t),
\end{align}
where the constraint, $\widetilde{C}1$ is considered as implicit and $\lambda_{ij}$ is the Lagrangian multiplier to the constraint, $\widetilde{C}3$. Using~\eqref{eq:lagrangian}, Karush–Kuhn–Tucker (KKT) conditions can be expressed as
\begin{subequations}\label{eq:KKT_cond}
	\begin{align}\label{eq:KKT_1}
		&\frac{\partial \mathcal{L}}{\partial t}=1-\sum_{i,j}\lambda_{ij}=0\Rightarrow \sum_{i,j}\lambda_{ij}=1,\\
		&\frac{\partial \mathcal{L}}{\partial x}=\sum_{i,j}\lambda_{ij}\frac{\partial d_{ij}}{\partial x}=0,\\\label{eq:KKT_3}
		&\lambda_{ij}(d_{ij}-t)=0, \forall i,j.
	\end{align}
\end{subequations}
Now from~\eqref{eq:KKT_cond}, the optimal location, $x'^*$ can be obtained using Lemma~\ref{lemma4} as described below.
\begin{lemma}\label{lemma4}
	Using the KKT conditions in~\eqref{eq:KKT_cond}, the optimal location, $x'^*$ for (P5) is given by
	\begin{align}\label{eq:glo_opt_loc}
		x'^*=\frac{D}{2}-(N-1)a.
	\end{align}
\end{lemma}
\begin{proof}
	From the KKT condition in~\eqref{eq:KKT_3}, if there are more than one Lagrangian multipliers, for instance, $\lambda_{ij},\lambda_{i^{'}j^{'}}>0$, then the corresponding path lengths, $d_{ij}=d_{i^{'}j^{'}}=t$. However, in the IRS panel, the elements are arranged in a grid, therefore, no two path lengths can be equal. So, only one Lagrangian multiplier will be $>0$ and the remaining will be $=0$. Also, from~\eqref{eq:KKT_1}, the nonzero Lagrangian multiplier is $=1$ which is associated with the maximum path length. If we follow Fig.~\ref{fig:sys_topo_multi}, for the middle axis of the panel $<D/2$, $d_{ij}$ is maximum for $i=M-1$ and $j=0$, else for $>D/2$, it is maximum for $i=M-1$ and $j=N-1$. Therefore, if we consider the former case, the search region for the location, $x'$ is $[0,D/2-(N-1)a]$ in order to minimize $d_{M-1,0}$ in $x'$, using~\eqref{eq:d_exp_multi}. Now, to determine it, in general, the first and second derivatives of $d_{ij}$ in $x'$ are given by
	\begin{subequations}
		\begin{align}\nonumber
			&\frac{\partial d_{ij}}{\partial x'}=\frac{x'+(2j+1)a}{\sqrt{(x'+(2j+1)a)^2+y'^2+(h'+(2i+1)a)^2}}\\\label{eq:d_first_der}
			&\hspace{11mm}-\frac{D-x'-(2j+1)a}{\sqrt{(D-x'-(2j+1)a)^2+y'^2+(h'+(2i+1)a)^2}},\\\nonumber
			&\frac{\partial^2 d_{ij}}{\partial x'^2}=\mathcal{Z}\left[\left(\frac{1}{(x'+(2j+1)a)^2+y'^2+(h'+(2i+1)a)^2}\right)^{\frac{1}{3}}\right.
			\\\label{eq:d_second_der}&\hspace{0mm}\left.+\left(\frac{1}{(D-x'-(2j+1)a)^2+y'^2+(h'+(2i+1)a)^2}\right)^{\frac{1}{3}}\right]>0.
		\end{align}
	\end{subequations}
	where $\mathcal{Z}=y'^2+(h'+(2i+1)a)^2$. From~\eqref{eq:d_second_der}, we find that $\frac{\partial^2 d_{ij}}{\partial x'^2}>0$, therefore, $d_{ij}$ is strictly convex in $x'$. Thus, using~\eqref{eq:d_first_der}, the global minimum value of $d_{ij}$ can be obtained by calculating the root of  $\frac{\partial d_{ij}}{\partial x'}=0$ which gives $x'=D/2-a$ for $i=M-1$ and $j=0$. But, the search region is $[0,D/2-(N-1)a]$ and $\frac{\partial d_{ij}}{\partial x'}$ is strictly decreasing in the region. Therefore, the minimum value of $d_{M-1,0}$ is obtained at $x'=D/2-(N-1)a$. Similarly, for the later case also, it can be shown that the minimum value of the maximum path length, $d_{M-1,N-1}$ is achieved at $x'=D/2-(N-1)a$. Thus, the globally optimal value of location $x'$ is given by~\eqref{eq:glo_opt_loc}.
\end{proof}
Using~\eqref{eq:d_exp_multi}, the minimum path length, $d_{ij}^*$ at $x'^*$ is given by $d_{ij}^*=\{d_{ij}|x'=x'^*\}$. Note that the optimal phase shift $\tilde{\theta}_{ij}$ in~\eqref{eq:opt_phase} is obtained at a given $d_{ij}$, the jointly optimal phase shift, $\theta_{ij}^*$ at $d_{ij}^*$ is obtained as
\begin{align}\label{eq:opt_phase_j}
	\theta_{ij}^*=\{\tilde{\theta}_{ij}|d_{ij}=d_{ij}^*\}.
\end{align}

Next, we validate the global optimality of the jointly optimal solution, $(x'^*,\{\theta_{ij}^*\})$ using the obtained numerical results discussed in the following section.
\color{black}

\section{Numerical Results and Analysis}

\subsection{Results Obtained for Single IRS Element}

Unless otherwise specified, the default value of the system parameters are set as \cite{wang06,kum03}: $P_t=2$ W, $D=10$ m, $h=4$ m, $s=8$ m$^2$, and $\lambda = 0.3278$ m, frequency, $f=915$ MHz.

\begin{figure}[!t]
		\centering  \includegraphics[width=2.6in]{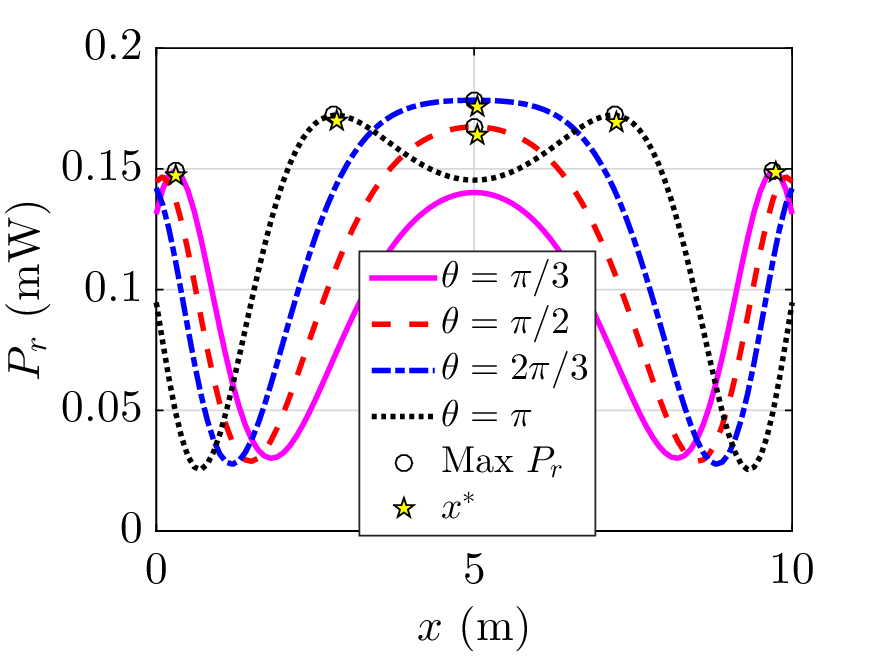}
		\caption{\small Design insights for optimal location of IRS for a given phase shift.}
		\label{fig:Opt_loc}
\end{figure}

\begin{figure}[!t]
		\centering  \includegraphics[width=2.6in]{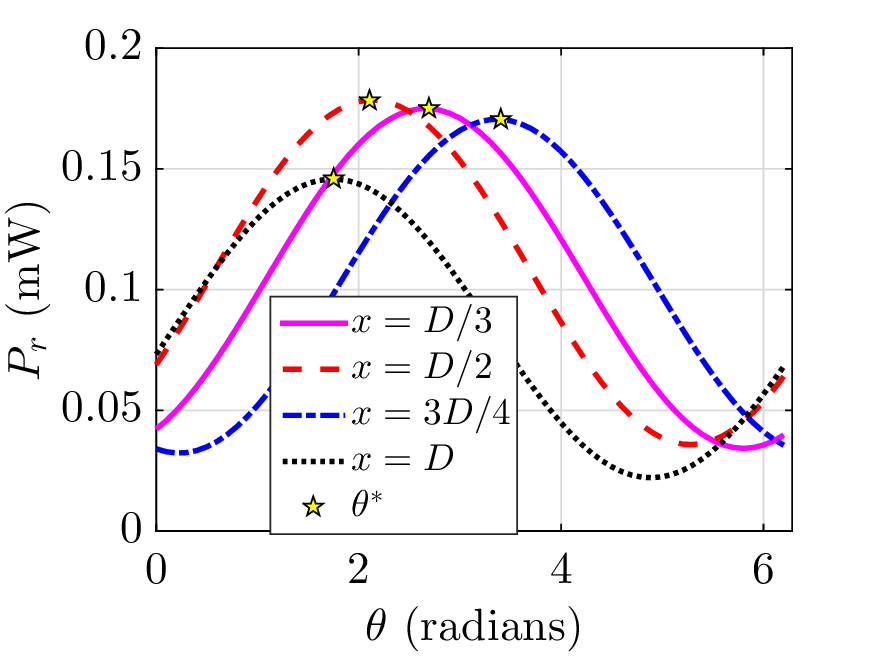}
		\caption{\small Using semi-adaptive scheme, the obtained optimal phase shift, $\theta^*$ for a given IRS location, $x$.}    \label{fig:opt_ang}
\end{figure}

Using Fig.~\ref{fig:Opt_loc}, we observe the variation of received power, $P_r$ with IRS location, $x$ for the given phase shift, $\theta$ at IRS. As the plot has multiple peak points, it depicts that problem (P1) is nonconvex. Here, the curves are symmetric about $x=D/2$, because, $P_r$ in~\eqref{eq:Pr_trac} varying in reflected path length, $d$ for a given phase shift, $\theta$. And, from~\eqref{eq:d_exp}, each value of $d$ satisfied by two values of $x$ which are symmetric about $x=D/2$. Further, Fig.~\ref{fig:Opt_loc} also describes the feasible optimal solution for problem (P1) where the obtained optimal solution using curve fitting polynomial equation in~\eqref{eq:Pr_poly} provides near globally optimal solution determined numerically using the maximum value of $P_r$ for a given $\theta$. If we compare the obtained optimal solutions for different phase shifts, for $\theta=\pi/3$, the optimal location of IRS is near to either AP, $\mathcal{A}$ or user, $\mathcal{U}$ as described in~\cite{wu07}. But, with $\theta$, optimal location, $x^*$ converges into $x=D/2$ and again diverges near to $\mathcal{A}$ or $\mathcal{U}$ for higher value of $\theta$. Note that for $\theta=\pi/3$ and $\pi$, we obtain the two solutions comprising equal maximum $P_r$, and one of them can be set as optimal location, $x^*$ as described in Section~\ref{sec:opt_loc}. 

Instead, in Fig.~\ref{fig:opt_ang}, we examine the obtained optimal phase shift using semi-adaptive scheme where the power, $P_r$ varies with $\theta$ for a given location, $x$. The curves in the plot follow the sinusoidal variations due to cosine factor in~\eqref{eq:Pr_trac} that dignifies (P2) as a nonconvex problem. However, as described in Section~\ref{sec:opt_phase}, over the restricted range of $\theta$, the problem becomes convex with unique globally optimal solution as shown in Fig.~\ref{fig:opt_ang}. Besides, with $x$, the optimal phase shift, $\theta^*$ also increases and maximum $P_r$ is achieved at $\theta^*=2.1$ Radians when the IRS is located at $x=D/2$.   

\begin{figure}[!t]
	
		\centering  \includegraphics[width=2.6in]{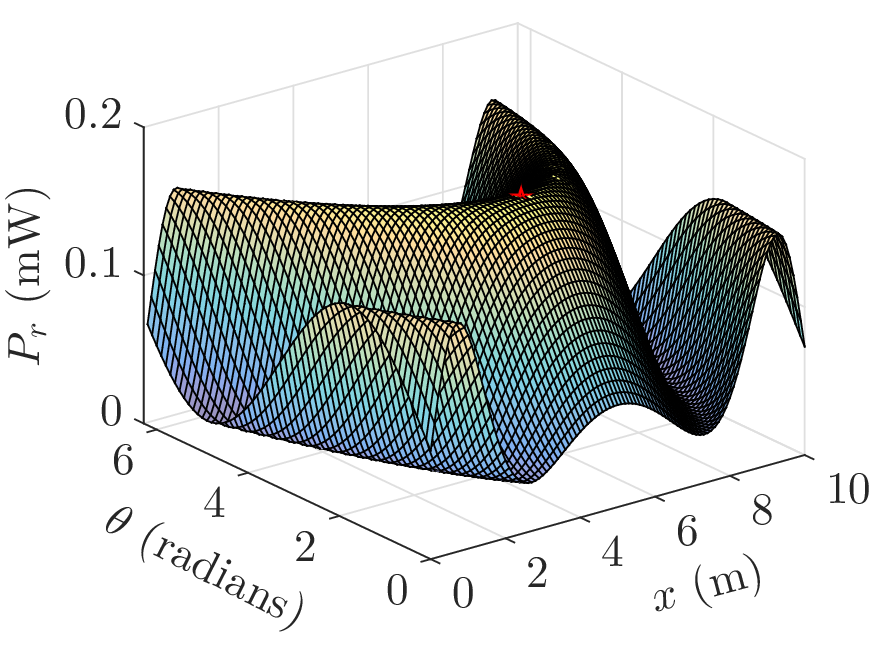}
		\caption{\small Insights on the joint variation of $P_r$ with $x$ and $\theta$.}
		\label{fig:3D_plot}
\end{figure}
	
\begin{figure}[!t]
		\centering  \includegraphics[width=2.8in]{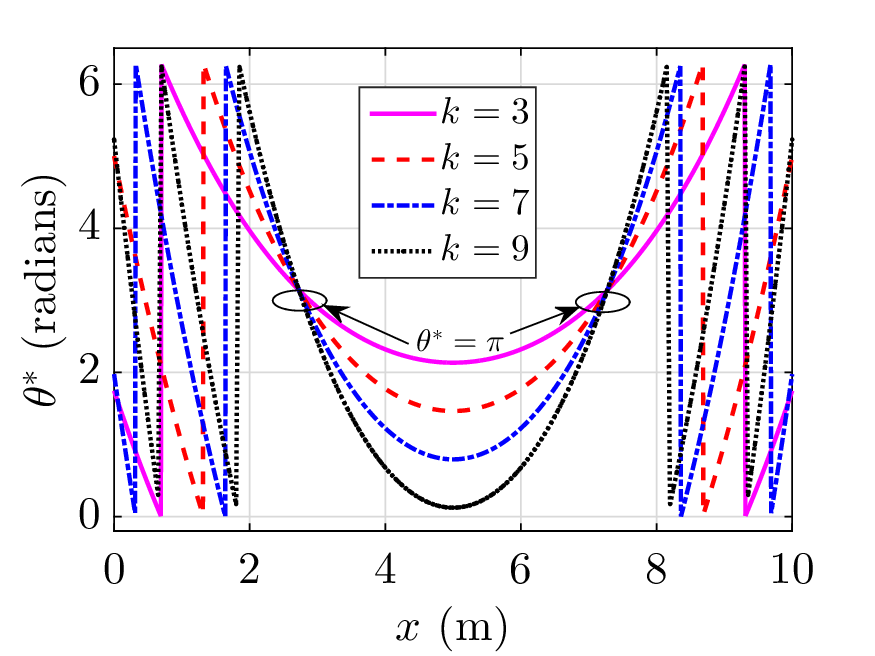}
		\caption{\small For a given wavenumber $k$, the variation of optimal phase shift, $\theta^*$ with location, $x$.}    \label{fig:opt_ang_loc}
\end{figure}

The variation of $P_r$ with $\theta$ and $x$ is shown in Fig.~\ref{fig:3D_plot} which depicts that the joint optimization problem is nonconvex due to multiple peak points. Also, the variation is in sinusoidal form due to cosine factor in the expression of $P_r$ in~\eqref{eq:Pr_trac}. Besides, variation of $P_r$ in $x$ ($\theta$) for a fixed $\theta$ ($x$) is not unimodal, so, the problems for semi adaptive schemes are nonconvex too as described using Figs.~\ref{fig:Opt_loc} and \ref{fig:opt_ang}. For the given setting of system parameters, the jointly optimal solution is $(x^*=5.13 \text{ m}, \theta^*=2.21 \text{ radians})$ which provides the maximum $P_r=0.18$ mW. 

Fig.~\ref{fig:opt_ang_loc} describes the variation of optimal phase shift, $\theta^*$ with location, $x$ for different values of wavenumber, $k$ which is obtained using the semi-adaptive scheme as described in Section~\ref{sec:opt_phase}. Here, the value of phase shift is concealed withing the range $[0,2\pi]$ using the modulo $2\pi$ operation. It can be observed that $\theta^*$ decreases with $x$ until $x<5$ m, otherwise, it always increases. For, $x\in(0\text{ m},2\text{ m})\cup (8\text{ m},10\text{ m})$, the curves switches at $\theta^*=0$ and $2\pi$ radians as the rate of change of $\theta^*$ is high over the range of $x$. Also, all the curves meet at $\theta^*=\pi$ as at corresponding location, the factor $(d-D)=\pi$ and for $k\in\{3,5,7,9\}$, $\theta^*=k(d-D)$ is odd integral multiple of $\pi$ which is always $\pi$ after modulo $2\pi$ operation. For a given $k$, $\theta^*$ gets its minimum value at $x=5$ m and with $k$, the minimum value decreases.

\subsection{Results Obtained for Multiple IRS Element}

Here, we investigate the analysis for multiple IRS elements. For the obtained numerical results, unless specified, the default system parameters are set as~\cite{moh08,ozd09}: $P_t=10$ W, $D=100$ m, $y=0.5$ m, $h=25$ m, $\lambda=0.12$ m, $a=0.0075$ m, $\Gamma = 0.5$, and $M=N=20$.

\begin{figure}[!t] 
	\centering  \includegraphics[width=3.4in]{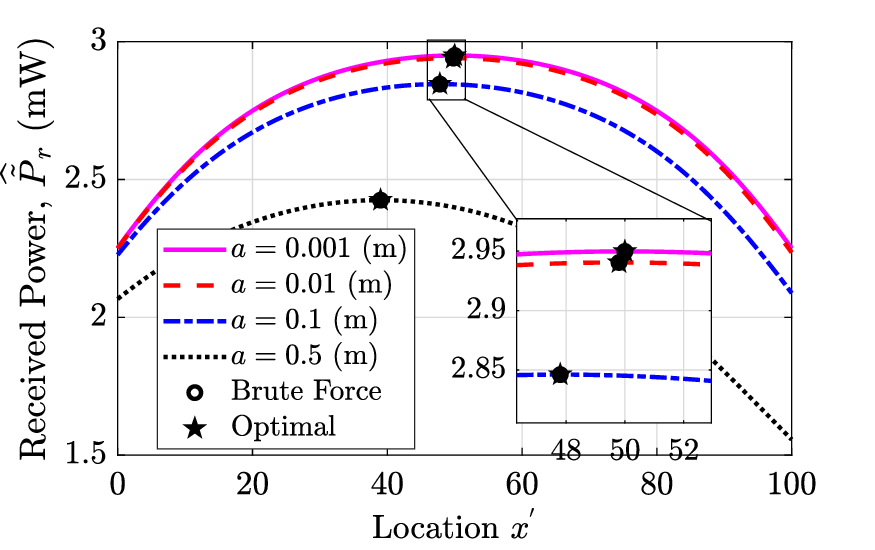}
	\caption{\small Insights on the variation of power, $\widehat{\widetilde{P}}_r$ in~\eqref{eq:rx_power_opt_multi} with size, $a$ and validation of optimal solution, $(x^*, \{\theta_{ij}^*\})$ in~\eqref{eq:glo_opt_loc} and \eqref{eq:opt_phase_j}.}    
	\label{fig:val_joint}\vspace{-4mm}
\end{figure}

Using Fig.~\ref{fig:val_joint}, we validate the global optimality of the obtained jointly optimal solution in~\eqref{eq:glo_opt_loc} and \eqref{eq:opt_phase_j}. In this figure, the received power, $\widehat{P}_{r}$ (cf. \eqref{eq:glo_opt_loc}) at the optimal phase shifts, $\{\tilde{\theta}_{ij}\}$ (cf.~\eqref{eq:opt_phase}) is varying with $x$ for different sizes, $a$ of the elements in the IRS panel. It shows that the variation is unimodal in nature and obtained solution to maximize the received power is globally optimal. After comparison of the performance of the obtained jointly optimal solution, $(x^*, \{\theta_{ij}^*\})$ against brute-force search, we find that both yield almost the same performance. Further, with the increment in $a$, the optimal location, $x^*$ of the IRS shifts left, towards the BS, $\mathcal{B}$, because the leftmost point, $x$ moves leftward about the middle point, $D/2$. Also, the power, $\widehat{\widetilde{P}}_r$ decreases exponentially because the shortest average path length via the reflection occurs when the elements are located at $D/2$ but with size, $a$ the elements disperse around $D/2$ which increases the average path length. 


\begin{figure}[!t] 
	\centering  \includegraphics[width=3.4in]{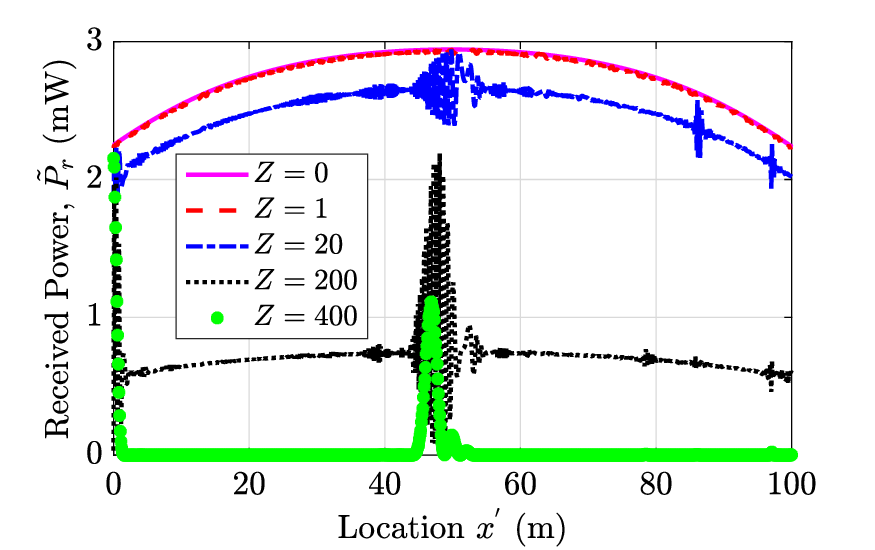}
	\caption{\small Influence on the received power, $\widetilde{P}_r$ in~\eqref{eq:Pr_trac_multi} due to different phase shifts at the IRS elements.}    
	\label{fig:Pr_phase}\vspace{-2mm}
\end{figure}

As discoursed in Section~\ref{sec:prob_form}, the amount of received power at $\mathcal{U}$ depends on the constructive and destructive interferences. To realize it, we have plotted the variation of $P_r$ in~\eqref{eq:Pr_trac} with location, $x$ for different values of phase shifts at IRS elements as shown in Fig.~\ref{fig:Pr_phase}. Here, $Z$ denotes the number of elements following the fixed phase shift, $2\pi$, whereas the remaining $(400-Z)$ elements have the optimal phase shifts, $\tilde{\theta}_{ij}$. When all the elements are set at optimal phase shifts $(Z=0)$ the variation of $P_r$ is unimodal with $x$, but if even a single element has a fixed phase shift $(Z=1)$, the variation is non-unimodal and $P_r$ deteriorates. Furthermore, for higher value of $Z$, the destructive interference at $\mathcal{U}$ is dominant and for $Z=400$, $P_r$ is almost zero. Besides, for $Z>20$, all the curves have spikes at around $x=D/2$ due to the high swing between constructive and destructive interference. 

\begin{figure}[!t] 
	\centering  \includegraphics[width=3.4in]{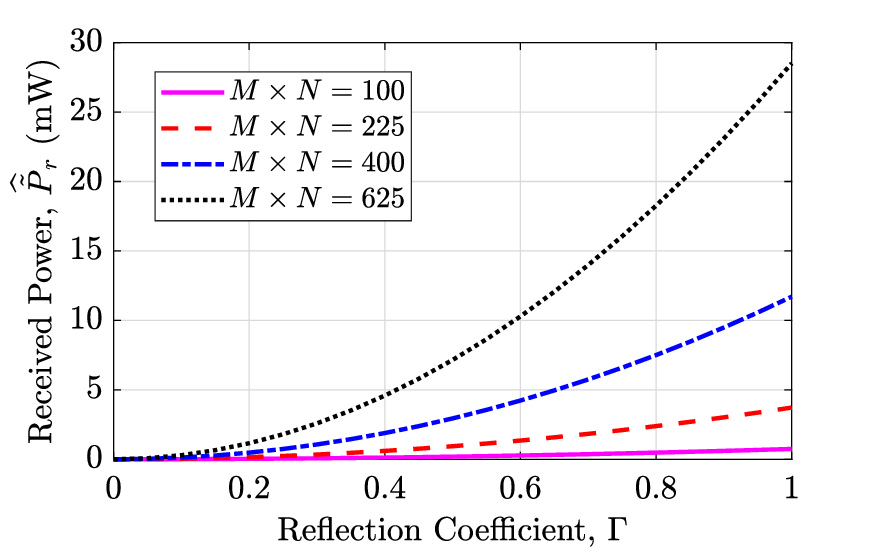}
	\caption{\small Impact on received power, $\widehat{\widetilde{P}}_r$ due to the number of elements and the reflection coefficient.}    
	\label{fig:Pr_Gamma}\vspace{-4mm}
\end{figure}

Using Fig.~\ref{fig:Pr_Gamma}, we investigate the influence of the reflection coefficient, $\Gamma$ and the number of elements, $M\times N$ on the received power, $\widehat{P}_r$. The figure depicts the variation of $\widehat{\widetilde{P}}_r$ with $\Gamma$ for different value of $M\times N$. Here, we can observe that for a given $M\times N$, $\widehat{\widetilde{P}}_r$ increases quadratically with $\Gamma$, because, $\widehat{\widetilde{P}}_r$ in~\eqref{eq:rx_power_opt_multi} is a polynomial function of $\Gamma$ of degree $2$. Besides, for a given $\Gamma$, $\widehat{P}_r$ significantly increases with number of elements, $M\times N$ due to increment in number of constructive interference of the reflected signals at the user, $\mathcal{U}$. For instance, the average received powers for $M\times N=1$ (in~\cite{jyo10}) and $100$ elements IRS are $0.2$ $\mu$W and $0.3$ mW, respectively.

\begin{figure}[!t] 
	\centering  \includegraphics[width=3.4in]{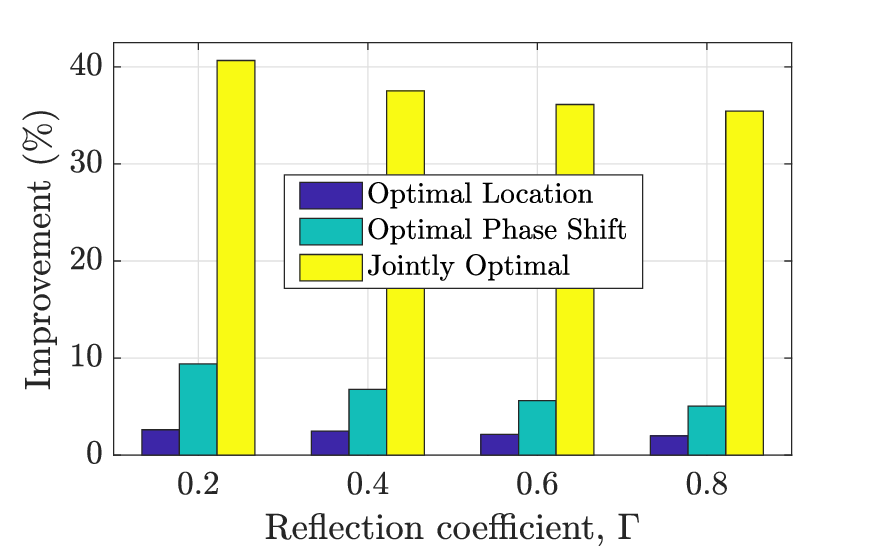}
	\caption{\small Performance comparison of different optimization schemes against the benchmark in~\cite{you11}.}    
	\label{fig:perf_comp}\vspace{-4mm}
\end{figure}

In Fig.~\ref{fig:perf_comp}, the different optimization schemes are compared in terms of received power against the closest competitive scheme presented in ~\cite{you11} where IRS location is set near to BS, i.e., $x=0$ and phase shift, $\theta_{ij}=2\pi$ $\forall i,j$. For the suboptimal schemes, optimal location and optimal phase shift, $x$ and $\{\theta_{ij}\}$ are optimized while keeping $\theta_{ij}=2\pi$; $\forall i,j$ and $x=0$, respectively. The figure shows the percentage performance improvement of the optimization schemes at different $\Gamma$ using a bar diagram. It can be observed that the joint optimization scheme gives significantly high performance enhancement than the suboptimal schemes. Also, optimal phase shift scheme provides always better performance than optimal location because, it plays the dominant role in the constructive interference of the reflected signals. On average, the optimal location, optimal phase shift, and the joint optimization provide $2.30\%$, $6.71\%$, and $37.44\%$, respectively.

\section{Concluding Remarks}

This paper elegantly analyzes the enhancement of received power in an IRS-assisted communication system by jointly optimizing the phase shifts of IRS elements and its location. The numerical investigation demonstrates the global optimality of the obtained solution, highlighting the critical influence of phase shifts on the interference of received signal components. The research shows that optimizing even a single IRS element's phase shift significantly impacts the non-unimodal variation of received power with IRS location, affecting overall performance. Moreover, increasing the number of IRS elements substantially improves received power. The proposed joint optimization scheme exhibits an impressive $37\%$ performance improvement compared to the closest competitive scheme. In conclusion, this work advances IRS technology's potential for beyond B5G and 6G networks, providing valuable insights into efficient and reliable wireless communication and paving the way for exciting innovations in this field.\\

\appendices
\setcounter{equation}{0}
\setcounter{figure}{0}
\renewcommand{\theequation}{A.\arabic{equation}}
\renewcommand{\thefigure}{A.\arabic{figure}} 
\enlargethispage*{2\baselineskip}
\section{}\label{App_A}

From~\eqref{eq:rx_pow}, the received power, $P_r$ at user, $\mathcal{U}$ is:
\begin{align}\nonumber
	\textstyle P_{r}&=\textstyle {P_{t}}\left( \frac{\lambda}{4\pi} \right) ^2\left| \frac{1}{D}e^{-jkD}+\frac{\Gamma}{d}e^{j\theta}e^{-jkd} \right|^2
	\\\nonumber
	&= \textstyle {P_{t}}\left( \frac{\lambda}{4\pi} \right) ^2\left| \frac{\cos \left( kD \right)
		-j\sin \left( kD \right)}{D}+\frac{\Gamma}{d}\left( \cos \left( kd-\theta \right)\right.\right.\\\nonumber
	&\left.\textstyle -j\sin \left( kd-\theta \right) \right) \bigg|^2
	\\\nonumber
	&=\textstyle {P_{t}}\left( \frac{\lambda}{4\pi} \right) ^2\left| \frac{\cos \left( kD \right)}{D}+\frac{\Gamma \cos \left( kd-\theta \right)}{d}\right.
	\\\label{eq:AppA1}
	&\left. \textstyle -j\left( \frac{j\sin \left( kD \right)}{D}+\frac{\Gamma \sin \left( kd-\theta \right)}{d} \right) \right|^2
\end{align}

Further, \eqref{eq:AppA1} can be simplified as:
\begin{align}\nonumber
	\textstyle P_{r}&=\textstyle {P_{t}}\left( \frac{\lambda}{4\pi} \right) ^2\left[ \left( \frac{\cos \left( kD \right)}{D}+\frac{\Gamma}{d}\cos \left( kd-\theta \right) \right) ^2\right.
	\\\nonumber
	&\left.\textstyle +\left( \frac{\sin \left( kD \right)}{D}+\frac{\Gamma \sin \left( kd-\theta \right)}{d} \right) ^2 \right]
	\\\nonumber
	&=\textstyle {P_{t}}\left( \frac{\lambda}{4\pi} \right) ^2\left[
	\frac{\cos ^2\left( kD \right)}{D^2}+\frac{\Gamma ^2}{d^2}\cos ^2\left( kd-\theta \right) \right.
	\\\nonumber
	&\left.\textstyle +\frac{2\Gamma}{dD}\cos \left( kD \right) \cos \left( kd-\theta \right) +\frac{\sin ^2\left( kD \right)}{D^2}\right.
	\\\nonumber
	&\left.\textstyle \frac{\tau ^2\sin ^2\left( kd-\theta \right)}{d^2}+\frac{2\Gamma}{dD}\sin \left( kD \right) \sin \left( kd-\theta \right) \right]
	\\\nonumber
	&=\textstyle {P_{t}}\left( \frac{\lambda}{4\pi} \right) ^2\left( \frac{1}{D^2}+\frac{\Gamma ^2}{d^2}+\frac{2\Gamma}{dD}\left( \cos \left( kD \right) \cos \left( kd-\theta \right) \right.\right.
	\\\nonumber
	&\left.\left.\textstyle +\sin \left( kD \right) \sin \left( kd-\theta \right) \right) \right)
	\\
	\!\!\!\!\!&=\textstyle {P_{t}}\left( \frac{\lambda}{4\pi} \right) ^2\hspace{-1mm}\left(\hspace{-0.5mm} \frac{1}{D^2}\hspace{-0.5mm}+\hspace{-0.5mm}\frac{\Gamma ^2}{d^2}\hspace{-0.5mm}+\hspace{-0.5mm}\frac{2\Gamma}{dD}\cos \left( kd-\theta -kD \right)\!\! \right)
\end{align}

\section{Proof of Lemma~\ref{lemma2}}\label{App_B}
To obtain the received power expression using~\eqref{eq:rx_pow_multi} in more tractable form, we adopt the induction method where the received power expressions for the reflection from one and two elements of the IRS are obtained as follows. If we consider the reflection from the $ij$th element of the IRS panel, as obtained in [A.2], the received power $\widetilde{P}_1$ at user $\mathcal{U}$ due to multi-elements IRS is given by

\begin{align}\label{eq:app_rx_1}
	\!\!\!\!\!\textstyle \widetilde{P}_{r,1} = \textstyle {P_{t}}\hspace{-0.5mm}\left(\hspace{-0.5mm} \frac{\lambda}{4\pi}\hspace{-1mm} \right) ^2\hspace{-1mm}\left(\hspace{-1mm} \frac{1}{D^2}\hspace{-0.5mm}+\hspace{-0.5mm}\frac{\Gamma ^2}{d_{ij}^2}\hspace{-0.5mm}+\hspace{-0.5mm}\frac{2\Gamma}{d_{ij} D}\cos \left( kd_{ij} -\theta_{ij} -kD \right)\!\! \right).\!\!\!\!
\end{align}

Further, using~\eqref{eq:rx_pow_multi}, for the reflection from the two elements, $ij$th and $i(j+1)$th, the received power, $\widetilde{P}_{r,2}$ is expressed as

\begin{align}\nonumber
	&\textstyle \widetilde{P}_{r,2}={P_{t}}\left( \frac{\lambda}{4\pi} \right) ^2\Big| \frac{1}{D}e^{-jkD}+\frac{\Gamma}{d_{ij}}e^{j(\theta_{ij}-kd_{ij})}+\frac{\Gamma}{d_{i(j+1)}}\\\nonumber&\hspace{10mm}\textstyle \times e^{j(\theta_{i(j+1)}-kd_{i(j+1)})}\Big|^2
	\\\nonumber
	&\textstyle=\hspace{-0.5mm}{P_{t}}\hspace{-0.5mm}\left( \frac{\lambda}{4\pi} \hspace{-0.5mm}\right)^2\hspace{-0.5mm}\bigg| \hspace{-0.5mm}\bigg(\hspace{-0.5mm}\frac{\cos \left( kD \right)}{D}\hspace{-0.5mm}+\hspace{-0.5mm}\frac{\Gamma \cos \left( \theta_{ij}-kd_{ij} \right)}{d_{ij}}+\frac{\Gamma \cos \left( \theta_{i(j+1)}-kd_{i(j+1)} \right)}{d_{i(j+1)}}\hspace{-0.5mm}\bigg)
	\\\nonumber
	&\left.\textstyle\hspace{4mm}+j\left( \frac{-\sin \left( kD \right)}{D}+\frac{\Gamma \sin \left( \theta_{ij}-kd_{ij} \right)}{d_{ij}}+\frac{\Gamma \sin \left( \theta_{i(j+1)}-kd_{i(j+1)} \right)}{d_{i(j+1)}} \right) \right|^2\\\nonumber
	&\textstyle={P_{t}}\left( \frac{\lambda}{4\pi} \right) ^2\hspace{-1mm}\left(\hspace{-0.5mm} \frac{1}{D^2}\hspace{-0.5mm}+\hspace{-0.5mm}\frac{\Gamma ^2}{d_{ij}^2}\hspace{-0.5mm}+\hspace{-0.5mm}\frac{\Gamma ^2}{d_{i(j+1)}^2}\hspace{-0.5mm}+\hspace{-0.5mm}\frac{2\Gamma}{d_{ij}D}\cos \left( kd_{ij}-\theta_{ij} -kD \right)\right.\\\nonumber\label{eq:App_rx_2}&\left.\textstyle\hspace{4mm}+\frac{2\Gamma}{d_{i(j+1)}D}\cos \left( kd_{i(j+1)}-\theta_{i(j+1)} -kD \right)+\hspace{-0.5mm}\frac{2\Gamma^2}{d_{ij}d_{i(j+1)}}\right.\\&\left.\hspace{4mm}\times\cos \left( \theta_{ij}-kd_{ij}-\theta_{i(j+1)} +kd_{i(j+1
		)} \right)\right).
\end{align}

Thus, using~\eqref{eq:app_rx_1} and \eqref{eq:App_rx_2}, recursively, the received power, $\widetilde{P}_r$ in~\eqref{eq:rx_pow_multi} can be expressed as in~\eqref{eq:Pr_trac_multi}.

\makeatletter
\renewenvironment{thebibliography}[1]{%
	\@xp\section\@xp*\@xp{\refname}%
	\normalfont\footnotesize\labelsep .5em\relax
	\renewcommand\theenumiv{\arabic{enumiv}}\let\p@enumiv\@empty
	\vspace*{-1pt}
	\list{\@biblabel{\theenumiv}}{\settowidth\labelwidth{\@biblabel{#1}}%
		\leftmargin\labelwidth \advance\leftmargin\labelsep
		\usecounter{enumiv}}%
	\sloppy \clubpenalty\@M \widowpenalty\clubpenalty
	\sfcode`\.=\@m
}{%
	\def\@noitemerr{\@latex@warning{Empty `thebibliography' environment}}%
	\endlist
}
\makeatother     
\bibliographystyle{IEEEtran}
\bibliography{references_COMML}

\end{document}